%% file: main.tex
\documentclass[11pt]{article}
\usepackage[scale=0.75,hmarginratio=1:1,vmarginratio=2:3,letterpaper]{geometry}
\usepackage{inputenc}
\usepackage{mdframed,adjustbox}
\usepackage{booktabs}
\usepackage[numbers]{natbib}

\usepackage{amsmath,amssymb,amsthm}
\newtheorem{theorem}{Theorem}

\newtheorem{proposition}[theorem]{Proposition}
\theoremstyle{definition}

\newcommand{\pv}{\mathbf{p}}
\newcommand{\qv}{\mathbf{q}}

\input{defns_ryu}

\title{An Information-Theoretic Proof of the Kac--Bernstein Theorem}
\author{J. Jon Ryu and Young-Han Kim\vspace{.3em}\\
Department of Electrical and Computer Engineering, UC San Diego\vspace{.3em}\\
\texttt{\{jongharyu,yhk\}@ucsd.edu}}
\begin{document}
\maketitle
\begin{abstract}
A short, information-theoretic proof of the Kac--Bernstein theorem, which is stated as follows, is presented: For any independent random variables $X$ and $Y$, if $X+Y$ and $X-Y$ are independent, then $X$ and $Y$ are normally distributed.
\end{abstract}

\section{Introduction}
The following statement is arguably the simplest characterization of normal distributions based on the independence of linear statistics.
\begin{theorem}[Kac--Bernstein~\cite{Kac1939,Bernstein1941}]
\label{thm:Kac_Bernstein}
Let $X$ and $Y$ be independent random variables with finite variances.
If $X+Y$ and $X-Y$ are independent, then $X$ and $Y$ are normally distributed.
\end{theorem}

The most standard proof known in the literature is based on the characteristic function, see, \eg \citep{Kagan--Linnik--Rao1973,Wlodzimierz1995}.
In this note, we provide an alternative proof of this statement using tools from information theory.

% \section{Main results}
% \subsection{Proof of Theorem~\ref{thm:Kac_Bernstein}}
\section{Preliminaries}
For a random variable $\Xv\in\Real^d$ with density $p(\xv)$, the differential entropy is defined as
\[
h(\Xv)\defeq \int p(\xv)\ln\frac{1}{p(\xv)}\diff \xv.
\]

\begin{theorem}[Entropy power inequality~\cite{Shannon1948,Stam1959,Cover--Thomas2006}]
\label{thm:EPI}
For any independent random vectors $\Uv$ and $\Vv$ in $\Real^d$, we have 
\[
e^{\frac{2}{d}h(\Uv+\Vv)} \ge e^{\frac{2}{d}h(\Uv)}+e^{\frac{2}{d}h(\Vv)},
\]
where the equality holds if and only if $\Uv$ and $\Vv$ are normal with proportional covariance matrices.
\end{theorem}

\noindent We will invoke this inequality with $d=1$: for any independent random variables $U$ and $V$,
\[
e^{2h(U+V)} \ge e^{2h(U)} + e^{2h(V)}.
\]

Further, note that
\begin{proposition}
\label{prop:scale}
$h(A\Xv)=h(\Xv)+\ln|\det(A)|$ for any nonsingular matrix $A\in\Real^{d\times d}$.
\end{proposition}

\section{Proof of Theorem~\ref{thm:Kac_Bernstein}}
We are now ready to prove Theorem~\ref{thm:Kac_Bernstein}.
% \begin{proof}[Proof of Theorem~\ref{thm:Kac_Bernstein}]
Without loss of generality, assume that $X$ and $Y$ have mean zero and variance one.
We apply the entropy power inequality (Theorem~\ref{thm:EPI}) with $(U,V)=(X,Y)$ and $(U,V)=(X,-Y)$:
\begin{align*}
% \text{(LHS)}
e^{2h(X+Y)}\ge e^{2h(X)} + e^{2h(Y)},\\
e^{2h(X-Y)}\ge e^{2h(X)} + e^{2h(Y)}.
\end{align*}
Taking a product with square root, we have
\begin{align}
e^{h(X+Y)+h(X-Y)} \ge e^{2h(X)} + e^{2h(Y)} \ge 2e^{h(X)+h(Y)},
\label{eq:interim}
\end{align}
where the last inequality follows from AM-GM inequality.
Since $X+Y$ and $X-Y$ are independent, using Proposition~\ref{prop:scale}, we have
\begin{align*}
h(X+Y)+h(X-Y) &= h(X+Y,X-Y) \\
&= h\Bigl(\begin{bmatrix}1 & 1 \\ 1 & -1\end{bmatrix}\begin{bmatrix}X \\ Y\end{bmatrix}\Bigr) \\
&=  h(X,Y) + \log 2\\
&=h(X)+h(Y)+\log 2.
\end{align*}
Therefore, we have
\[
e^{h(X+Y)+h(X-Y)}=2e^{h(X)+h(Y)},
\]
which implies that the inequalities in \eqref{eq:interim} must hold with equality.
In particular, from the equality condition of the entropy power inequality, $X$ and $Y$ must be normal.\qed
% \end{proof}

\section{An Extension}
By the same argument in the previous proof, we can prove a quick generalization.

\begin{theorem}
Let $\Xv=(X_1,\ldots,X_n)$ be a random vector whose entries are mutually independent with finite variances. 
For a matrix $Q\in\Real^{n\times n}$ with nonzero entries,
let $\Wv=(W_1,\ldots,W_n)=Q\Xv$.
If $W_1,\ldots,W_n$ are independent, then $X_1,\ldots,X_n$ are normally distributed.
\end{theorem}

\begin{proof}
Without loss of generality, assume $\E[X_1]=\ldots=\E[X_n]=0$ and $\Var(X_1)=\ldots=\Var(X_n)=1$. 
Since $W_1,\ldots,W_n$ are independent, the matrix $Q$ must be orthogonal since $\Cov(W_i,W_j)=\E[W_iW_j]=\d_{ij}$. 
Let $\qv_i$ be the $i$-th column of $Q$.
We will apply the EPI to each component of $\Wv$:
\begin{align*}
e^{2h(W_i)} 
&\ge e^{2h(q_{i1} X_1)} + \cdots + e^{2h(q_{in} X_n)}
\\&
= q_{i1}^2 e^{2h(X_1)}+\cdots+q_{in}^2 e^{2h(X_n)}
\\&
\stackrel{(a)}{\ge} e^{2 (q_{i1}^2 h(X_1)+\cdots+q_{in}^2 h(X_n))}.
\end{align*}
Here, $(a)$ follows from Jensen's inequality.
Taking the product of the inequalities for $i=1,\ldots,n$, we have
\begin{align*}
\text{(LHS)}
= e^{2(h(W_1) + \cdots + h(W_n))}
= e^{2h(\Wv)}
= e^{2h(Q\Xv)}
= e^{2h(\Xv)}
\end{align*}
and, since $\qv_1,\ldots,\qv_n$ are unit vectors,
\begin{align*}
\text{(RHS)}
= e^{2((q_{11}^2 + \cdots + q_{n1}^2) h(X_1) + \cdots + (q_{1n}^2 + \cdots + q_{nn}^2) h(X_n))}
= e^{2(h(X_1) + \cdots + h(X_n))}
= e^{2h(\Xv)}.
\end{align*}
Therefore, all inequalities must hold with equality.
Since all entries of $Q$ are nonzero, $W_1,\ldots,W_n$ must be normal, and thus so are $X_1,\ldots,X_n$.
\end{proof}

\section{Remarks}
There exist a few information-theoretic proofs of the Kac--Bernstein theorem in the literature. 
\citet{Itoh1970} proved a weaker statement using a rather complicated limit argument. 
Recently, \citet{Rioul2017} showed that the theorem immediately follows as a corollary from the equality condition of a certain form of reverse entropy power inequality. 
Albeit the idea might be identical in disguise, we believe that the argument in this paper is more direct.

Indeed, much stronger generalizations of Theorem~\ref{thm:Kac_Bernstein} exist; see \citep{Kagan--Linnik--Rao1973} for a comprehensive treatment of such statements.
One stated below, which was independently discovered by Darmois~\cite{Darmois1953} and Skitovich~\cite{Skitovitch1953}, is the most popular among them.
Remarkably, even for $n\ge 2$ independent random variables, only independence of two linear statistics implies normality of individual components. 
\begin{theorem}[Darmois--Skitovich~\cite{Darmois1953,Skitovitch1953}]
Let $\Xv=(X_1,\ldots,X_n)$ be a random vector whose entries are mutually independent with finite variances. 
Let $\pv_1,\pv_2\in\Real^n$ be vectors with nonzero entries. 
If $W_1=\pv_1^T\Xv$ and $W_2=\pv_2^T\Xv$ are independent, then $X_1,\ldots,X_n$ are normally distributed.
\end{theorem}

It is rather unclear if this strong statement can be proved by a similar information-theoretic argument in this paper.

\section*{Acknowledgments}
The authors appreciate Ioannis Kontoyiannis for pointing out existing information-theoretic proofs of the Kac--Bernstein theorem by \citet{Itoh1970} and \citet{Rioul2017}.

\bibliographystyle{plainnat}
\bibliography{ref}
\end{document}

%% file: defns_ryu.tex
\usepackage{xspace}
\usepackage{bbm}
\input{mathlig}
\usepackage{mathtools}
\usepackage{relsize}
\usepackage{mathrsfs}
\usepackage{dsfont}
\DeclarePairedDelimiterX{\inp}[2]{\langle}{\rangle}{#1, #2}
\makeatletter
\newcommand*\bigcdot{\mathpalette\bigcdot@{.5}}
\newcommand*\bigcdot@[2]{\mathbin{\vcenter{\hbox{\scalebox{#2}{$\m@th#1\bullet$}}}}}
\makeatother

\newcommand{\muspace}{\mspace{1mu}}

\DeclareRobustCommand{\scond}{\mathchoice{\muspace\vert\muspace}{\vert}{\vert}{\vert}}
\mathlig{|}{\scond}

\DeclareRobustCommand{\discint}{\mathchoice{\mspace{-1.5mu}:\mspace{-1.5mu}}{\mspace{-1.5mu}:\mspace{-1.5mu}}{:}{:}}
\mathlig{::}{\discint}

\def\Var{\mathop{\rm Var}\nolimits}%
\def\Cov{\mathop{\rm Cov}\nolimits}%
\newcommand{\Xv}{{\bf X}}

\newcommand{\Uv}{{\bf U}}
\newcommand{\Wv}{{\bf W}}
\newcommand{\Vv}{{\bf V}}

\newcommand{\xv}{{\bf x}}

\def\d{\delta}

\DeclareMathOperator\E{\mathsf{E}}

\newcommand\eg{e.g.,\xspace}

\def\textiid{i.i.d.\@\xspace}
\newcommand\iid{\ifmmode\text{ i.i.d. } \else \textiid \fi}

\newcommand{\Real}{\mathbb{R}}

\def\mathllap{\mathpalette\mathllapinternal}
\def\mathllapinternal#1#2{%
  \llap{$\mathsurround=0pt#1{#2}$}}

\def\clap#1{\hbox to 0pt{\hss#1\hss}}
\def\mathclap{\mathpalette\mathclapinternal}
\def\mathclapinternal#1#2{%
  \clap{$\mathsurround=0pt#1{#2}$}}

\let\oldstackrel\stackrel
\renewcommand{\stackrel}[2]{\oldstackrel{\mathclap{#1}}{#2}}

\DeclarePairedDelimiterX{\infdivx}[2]{(}{)}{%
  #1\;\delimsize\|\;#2%
}

\renewcommand{\hbar}{h\mathllap{\overline{\vphantom{h}\hphantom{\rule{4.6pt}{0pt}}}\mspace{0.77mu}}}

\catcode`~=11 % make LaTeX treat tilde (~) like a normal character
\newcommand{\urltilde}{\kern -.06em\lower -.06em\hbox{~}\kern .02em}
\catcode`~=13 % revert back to treating tilde (~) as an active character

\usepackage{xparse}

\newcommand*\diff{\mathop{}\!\mathrm{d}}

\let\oldpartial\partial
\renewcommand*{\partial}{\mathop{}\!\oldpartial}
\newcommand{\defeq}{\mathrel{\mathop{:}}=}

%% file: main.bbl
\begin{thebibliography}{11}
\providecommand{\natexlab}[1]{#1}
\providecommand{\url}[1]{\texttt{#1}}
\expandafter\ifx\csname urlstyle\endcsname\relax
  \providecommand{\doi}[1]{doi: #1}\else
  \providecommand{\doi}{doi: \begingroup \urlstyle{rm}\Url}\fi

\bibitem[Bernstein(1941)]{Bernstein1941}
SN~Bernstein.
\newblock On a property which characterizes a gaussian distribution.
\newblock \emph{Proc. Leningrad Polytech. Inst.}, 217\penalty0 (3):\penalty0
  21--22, 1941.

\bibitem[Bryc(1995)]{Wlodzimierz1995}
Włodzimierz Bryc.
\newblock \emph{Normal distribution: Characterizations with applications},
  volume 100 of \emph{Lecture Notes in Statistics}.
\newblock Springer, 1995.

\bibitem[Cover and Thomas(2006)]{Cover--Thomas2006}
Thomas~M Cover and Joy~A Thomas.
\newblock \emph{Elements of information theory}.
\newblock John Wiley \& Sons, 2006.

\bibitem[Darmois(1953)]{Darmois1953}
George Darmois.
\newblock Analyse g{\'e}n{\'e}rale des liaisons stochastiques: etude
  particuli{\`e}re de l'analyse factorielle lin{\'e}aire.
\newblock \emph{Rev. Inst. Int. Stat.}, pages 2--8, 1953.

\bibitem[Itoh(1970)]{Itoh1970}
Yoshiaki Itoh.
\newblock The information theoretic proof of {K}ac's theorem.
\newblock \emph{Proc. Jpn. Acad.}, 46\penalty0 (3):\penalty0 283--286, 1970.

\bibitem[Kac(1939)]{Kac1939}
Mark Kac.
\newblock On a characterization of the normal distribution.
\newblock \emph{Am. J. Math.}, 61\penalty0 (3):\penalty0 726--728, 1939.

\bibitem[Kagan et~al.(1973)Kagan, Linnik, and Rao]{Kagan--Linnik--Rao1973}
A.M. Kagan, Yu.~V. Linnik, and C.R. Rao.
\newblock \emph{Characterization problems in mathematical statistics}.
\newblock Wiley, New York, 1973.

\bibitem[Rioul(2017)]{Rioul2017}
Olivier Rioul.
\newblock Optimal transportation to the entropy-power inequality.
\newblock In \emph{Proc. Inf. Theory Appl. Workshop}, pages 1--5. IEEE, 2017.

\bibitem[Shannon(1948)]{Shannon1948}
Claude~E Shannon.
\newblock A mathematical theory of communication.
\newblock \emph{Bell Syst. Tech. J.}, 27\penalty0 (3):\penalty0 379--423, 1948.

\bibitem[Skitovitch(1953)]{Skitovitch1953}
V.~P. Skitovitch.
\newblock On a property of the normal distribution.
\newblock \emph{DAN SSSR}, 89:\penalty0 217--219, 1953.

\bibitem[Stam(1959)]{Stam1959}
Aart~J Stam.
\newblock Some inequalities satisfied by the quantities of information of
  {F}isher and {S}hannon.
\newblock \emph{Inf. Contr.}, 2\penalty0 (2):\penalty0 101--112, 1959.

\end{thebibliography}
